\documentclass{article} 
\usepackage{nips12submit_e,times}
\usepackage{graphicx}
\usepackage{enumerate}
\usepackage{algorithm}
\usepackage{algorithmic}
\usepackage{amsmath}
\usepackage{amssymb}
\usepackage{amsthm}
\allowdisplaybreaks[1]

\newtheorem{lem}{Lemma}

\newcommand{\eat}[1]{}
\DeclareMathOperator*{\argmin}{argmin}

\title{Symmetric Submodular Clustering with Actionable Constraint}

\author{
Amit Dhurandhar\\
Mathematical Sciences Dept.\\
IBM TJ Watson\\
Yorktown Heights, NY, USA \\
\texttt{adhuran@us.ibm.com} \\
\And
Karthik Gurumoorthy\thanks{This research work benefited from the support
of the AIRBUS Group Corporate Foundation Chair in Mathematics of Complex
Systems established in ICTS-TIFR.}\\
Intl. Center for Theoretical Sciences\\
Tata Inst. of Fundamental Research\\
Bangalore, Karnataka, India\\
\texttt{karthik.gurumoorthy@icts.res.in} \\
}

\nipsfinalcopy 

\begin{document}

\maketitle

\begin{abstract}
Clustering with submodular functions has been of interest over the last few years. Symmetric submodular functions are of particular interest as minimizing them is significantly more efficient and they include many commonly used functions in practice viz. graph cuts, mutual information. In this paper, we propose a novel constraint to make clustering actionable which is motivated by applications across multiple domains, and pose the problem of performing symmetric submodular clustering subject to this constraint. We see that obtaining a $k$ partition with approximation guarantees is a non-trivial task requiring further work.
\end{abstract}




\section{Introduction}
Given a dataset, finding homogeneous disjoint collections or clusters have applications in a wide variety of domains \cite{charubook,clustbook}. However, in many domains actions cannot be taken at an individual instance level but can only be taken based on homogeneity at predefined aggregate levels, which we refer to as groups. In these domains, unsupervised clustering may not provide actionable insight as instances from all the groups may be spread across multiple clusters with no cluster containing a sizable number of instances from any of the groups.

We thus define a new notion for a clustering to be useful called \emph{actionable clustering}. This notion is motivated by applications across multiple industries and domains. For example, many big businesses want to cluster their spend data \cite{icebe11} to find out areas of high/low spend and high/low non-compliance to be able to take the appropriate action. The appropriate action may be remedial in nature or they may want to reward certain entities. Given this, it is impractical for a business to take the appropriate action at the transaction level rather they may be able to put in place policies and processes at higher levels that respect their organizational structure such as the category (viz. marketing, human resources, IT, etc.) level. They would hope that the clustering will point towards one or more categories that they ought to target. This would require a high percentage of transactions corresponding to at least one category to lie in a single cluster. If the clustering is conducted in unsupervised fashion the transactions corresponding to the different categories may be spread across the various clusters rendering the clustering useless. However, it might have been the case that for a not much worse clustering\footnote{This is relative to the objective we are minimizing/maximizing.} most of the transactions in marketing would have landed in the same cluster, which would have made the clustering actionable and thus useful.
Note that the actionable entity may not just be based on a single attribute such as category but could be a combination of attributes such as category and business unit. A similar need can be seen in education where governing states might want to identify a school(s) under their jurisdiction that have been performing subpar/above expectation based on their students test scores, academic honors, athletic achievements, etc. in order to decide their funding levels. In this case too unsupervised clustering might not provide actionable clusters where most of the students in a particular school belong to a specific cluster, thus precluding the possibility of finding a consistent pattern at the school level and making the clustering unusable. Many such examples are seen in other domains viz. health-care and public policy, where decisions can be made only at a certain higher level of granularity and we want our clustering to respect this fact, even at the expense of obtaining a slightly worse clustering from the mathematical standpoint. From the above examples we can see that our definition for a clustering to be actionable is that there should exist at least one cluster which has a significant fraction $t$ of instances belonging to one of the groups we are interested in. In the above examples a group may correspond to a particular category or a specific school. Note that our constraint does not eliminate the possibility of having multiple groups being well represented in the clustering, rather it ensures that at least one such group is represented well enough.

The above exposition does not imply that the clusters should be homogeneous with respect to (w.r.t.) the groups as in supervised clustering \cite{super,thorsten}, where the groups could be considered as proxies for class labels, but rather a large fraction of instances belonging to some group should be present in some cluster. That cluster may very well have instances belonging to other groups. Moreover, a clustering which is excellent from the supervised perspective may not be feasible relative to our constraint, as each cluster may be homogeneous and contain only a single group but no cluster may contain at least $t$ fraction of the instances from any specific group.

Our definition of usefulness cannot be effectively captured in the constraint based or label based semi-supervised clustering frameworks \cite{wagstaff,semigrira,semibasu} either. The reason being that we do not know which $t$ fraction of the instances belonging to a group should be assigned to some cluster, so as to obtain an excellent clustering. We can of course randomly choose these instances and then perform semi-supervised clustering. However, we might have missed a different set of instances which if we had chosen as the $t$ fraction, would have resulted in a much better clustering.

With this new notion we do \emph{not} in any way imply that it covers all possibilities for a clustering to be useful but rather that it can lead to actionable clusterings in many applications. Moreover, none of the current frameworks or algorithms can effectively model our notion, which fosters the need for the design of new techniques. Note that variants such as performing clustering independently on each group would still require applying our constraint to make it actionable and would be more restrictive than our approach since, one would not be able to decipher if multiple groups can be well represented in the same cluster. This is important information as the corresponding organization can then design a better action plan based on the knowledge that these groups are similar.

We study the application of this constraint to symmetric submodular clustering. We do this since, symmetric submodular functions cover an interesting array of functions viz. graph cuts, mutual information, etc. commonly used in practice and (unconstrained) clustering with this class of functions has efficient polynomial time algorithms that provide globally optimal solutions for $k=2$ clusters and a better than 2 approximation guarantee for larger $k$ \cite{quey,queyext}. We explore the possibility of obtaining similar qualitative results in our constrained setting.

\section{Problem Statement}
Let $D$ denote the dataset of size $N$ we want to partition into $k$ clusters $\mathcal{C}=\{C_1,...,C_k\}$. Let $(f,D)$ represent a symmetric submodular system. Let $G=\{g_1,...,g_m\}$ denote the partition of $D$ into $m$ groups, with $g_s$ representing the smallest group. Let $|.|$ denote cardinality and $\lceil . \rceil$ denote ceiling. Given this and with $t\in [0,1]$ we want to solve the following problem:
\begin{equation}
\label{gobj}
\begin{split}
&\text{Actionable Symmetric Submodular Clustering}\\\\
&\argmin_{\mathcal{C}} \sum_{i=1}^{k} f(C_i) \text{ subject to: } \exists g\in G~~\exists C_i\in \mathcal{C}~~\text{such that}~~~ \frac{|C_i\bigcap g|}{|g|} \ge t
\end{split}
\end{equation}
where, $k\in \{1,...,N-\lceil t|g_s|\rceil+1\}$. For $k>N-\lceil t|g_s|\rceil+1$, it is impossible to create a $k$ partition and at the same time satisfy our constraint. In practice though, we usually desire only a few clusters for a concise interpretation of our data and thus feasibility is unlikely to be an issue.

\section{(Possible) Solutions}
In this section we explore two different strategies based on the unconstrained setting and some recent work \cite{heredit} with the hope of obtaining an algorithm that runs in polynomial time and has similar approximation guarantees for any feasible $k$.

Both the strategies arise from the fact that our constraint can be expressed as a hereditary family. A hereditary family $\mathcal{I}$ over $D$ is a collection of subsets of $D$ such that if a set is in the family, so are all of its subsets. In other words, $\mathcal{I}$ is \emph{closed} under inclusion. Hence, in our case we want to minimize the objective over all sets $S$, with $S\subseteq D$ and there existing a $g\in G$ such that $\frac{|S\bigcap g|}{|g|} \leq 1-t$. More succintly, our constraint corresponds to optimizing over the following hereditary family,
\begin{equation}
\mathcal{I} = \left\{ S \subseteq D: \exists g\in G \text{ with } \frac{|S\bigcap g|}{|g|} \leq 1-t \right\}.
\end{equation}

It is straightforward to verify that $\mathcal{I}$ respects the condition of an hereditary family since, for any feasible set $S$ defined as above, all its subsets also lead to feasible solutions.

\subsection{Strategy 1: Sequential Splitting}
Here we explore the first strategy which involves successively greedily splitting the intermediate partitions until we have a $k$ partition. Before we explore this alternative it is important to realize that the aforesaid problem can be equivalently reduced into a constraint problem on an \emph{a priori decided group} as follows. For every group $g_j$ set $\tilde{g}_{1j} = g_j$ and $\tilde{g}_{2j} = \bigcup\limits_{i \not= j} g_i$. Then the solution to
\begin{align}
\label{eq:twogroupproblem}
\argmin_{j} \left[\argmin_{\mathcal{C}} \sum_{i=1}^{k} f(C_i) \hspace{10pt}\text{subject to:}\hspace{10pt}\exists C_i\in \mathcal{C}~~\text{such that}~~~ \frac{|C_i\bigcap \tilde{g}_{1j}|}{\left|\tilde{g}_{1j}\right|} \ge t \right]
\end{align}
can be shown to correspond to the optimal actionable clustering partition for the general multi-group problem. This reduction allows us to focus only on the actionable clustering constraint mandated on a chosen group which when iterated---by choosing a different group $g_j$ every time---solves the multi-group case.

\subsubsection{Optimal solution for $k=2$}

By defining an appropriate hereditary family $\mathcal{I}$ as above we realize that the actionable requirement can be reformulated as an hereditary constraint and hence the first algorithm developed in \cite{heredit} provides the optimal solution for actionable clustering when $k=2$. This work handles the $k=2$ case and finds a non-trivial set $S \in  \mathcal{I}$ that minimizes $f$ using $O(N^3)$ function value oracle calls over a hereditary family. 

Let $S^{\ast} = \argmin\limits_{S \in \mathcal{I}} f(S)$ be the solution procured by invoking the method in \cite{heredit}. It then follows that the set $D-S^{\ast}$ satisfies the actionable clustering constraint for $g$ and furthermore, $f(D-S^{\ast}) = f(S^{\ast})$ as $f$ is symmetric. The partition $\{S^{\ast}, D-S^{\ast} \}$ is the optimal clustering solution.

\subsubsection{Extension to arbitrary $k$}
The natural follow-up question is whether we can design a fast optimal algorithm for an arbitrary $k \geq 2$ scenario. Before we move forward on this front it is worth reflecting back on the progress made in the relatively easier unconstrained setting. Though polynomial time algorithm exists for optimal multi-way partition for a symmetric sub-modular function \cite{Queyranne99}, the solution time grows very rapidly with $k$. Hence faster approximate algorithms are sought. The works in \cite{Queyranne99} and \cite{queyext} provide a fast $2-\frac{2}{k}$ approximation algorithm using $O(kN^3)$ function value oracle calls. The proof technique developed in the latter is easy to comprehend and is based on a greedy splitting algorithm (GSA) avoiding complicated structures like cut trees or principal partition used in \cite{Queyranne99}. GSA works iteratively whereby it constructs a partition $\mathcal{P}_{i+1}$ from $\mathcal{P}_i$ by first optimally splitting each $W \in \mathcal{P}_i$ into non-trivial sets $S$ and $W-S$ and then choosing that set $W_i$ with the least cost $f(S_i)+f(W_i-S_i)-f(W_i)$ as the candidate. The partition $\mathcal{P}_{i+1}$ is set to $(\mathcal{P}_{i}-\{W_i\})\bigcup \{S_i,W_i-S_i\}$. The reader may refer to \cite{queyext} for details.

Since GSA functions greedily where at each stage it only requires optimal partitioning of a set into \emph{two} subsets which is provably optimal even under the hereditary family constraints, can the GSA strategy be used in conjunction with the algorithm developed in \cite{heredit} to devise a fast algorithm with a similar $2-\frac{2}{k}$ approximation guarantee? Contrary to our expectation we explain with a counterexample that a straightforward implementation of GSA will \emph{not} fetch us the desired result.

\subsubsection{GSA  for Actionable Clustering}
We will fix some more notations before we describing the details of the algorithm. For a set $W \subseteq D$, let $|g(W) = g \bigcap W|$ denote the number of instances of the group $g$---for which the clustering output needs to be actionable---in $W$ and let $\mathcal{I}(W)$ denote its hereditary family defined as $S \subseteq W \in \mathcal{I}(W) \text{ \emph{iff} } |g(S)| \leq |g(W)|-\lceil t |g(D)| \rceil$. Also, let $c(W)$ denote the increment in cost incurred when $W$ is optimally split into non-trivial subsets $S$ and $W-S$, i.e, $c(W) = f(W-S) + f(S) - f(W)$. The algorithm employed to divide $W$ is contingent on a binary group constraint tag $\beta(W)$ that we attach to each set $W$. If $\beta(W)=0$ corresponding to the unconstrained scenario, $W$ is partitioned by invoking the method in \cite{quey}. When $\beta(W)=1$ relating to the constrained setting, the hereditary family algorithm in \cite{heredit} is used to split $W$ where $S$ is enforced to belong in $\mathcal{I}(W)$. Akin to GSA, the algorithm progresses iteratively with $W=D$ and $\beta(W) = 1$. The algorithm is spelled out in Table~\ref{tab:algorithm}. Observe that in every iteration $\beta(W)=1$ for \emph{exactly one} set $W \in \mathcal{P}_i$ as it suffices if one cluster satisfies the cardinality constraint for the clustering to be actionable for the group $g$. Following the analysis in \cite{queyext} and \cite{heredit} the running time of the proposed algorithm can also be shown to be $O(kN^3)$ function value oracle calls. But, can the algorithm provide any performance guarantee like the GSA? The answer turns out to be \emph{no} as elucidated below.
\begin{table}
\begin{centering}
\caption{GSA for actionable clustering}
\label{tab:algorithm}
\par\end{centering}
\centering{}
\begin{tabular}{|c|l|}
\hline 
1 & $\mathcal{P}_1 \leftarrow \{D\}$, $\beta(D) = 1$ \tabularnewline
2  & \textbf{for} $i = 1,\ldots,k-1$ \textbf{do}\tabularnewline
3  &\hspace{10pt}\textbf{for} each $W \in \mathcal{P}_i$ \textbf{do}\tabularnewline
4  &\hspace{20pt}\textbf{if} $\beta(W)=0$ \textbf{then}\tabularnewline
5  &\hspace{30pt}Partition $W$ into $S$ and $W-S$ by Queyranne's algorithm \cite{quey}. \tabularnewline
6  &\hspace{30pt}Set $\beta(S)=0$, $\beta(W-S)=0$ and $c(W) =  f(W-S) + f(S) - f(W)$.  \tabularnewline
7  &\hspace{20pt}\textbf{else} \tabularnewline
8  &\hspace{30pt}\textbf{if} $|W| = \lceil tN \rceil$ \textbf{then} set $c(W) = \infty$. \tabularnewline
9  &\hspace{30pt}\textbf{else} \tabularnewline
10  &\hspace{40pt}Define $\mathcal{I}(W)$ and partition $W$ into $S$ and $W-S$ by the hereditary\tabularnewline
& \hspace{40pt}family algorithm \cite{heredit} such that $S \in \mathcal{I}(W)$. \tabularnewline
11  &\hspace{40pt}Set $\beta(S)=0$, $\beta(W-S)=1$ and $c(W) =  f(W-S) + f(S) - f(W)$.  \tabularnewline
12  &\hspace{30pt}\textbf{end} /* end if */\tabularnewline
13  &\hspace{20pt}\textbf{end}  /* end if */\tabularnewline
14  &\hspace{10pt}\textbf{end}  /* end for */\tabularnewline
15  &\hspace{10pt}$(S_i,W_i) \leftarrow \argmin\limits_{W \in \mathcal{P}_i} c(W)$.\tabularnewline
16  &\hspace{10pt}$\mathcal{P}_{i+1}\leftarrow(\mathcal{P}_{i}-\{W_i\})\bigcup \{S_i,W_i-S_i\}$.\tabularnewline
17  &\textbf{end}  /* end for */\tabularnewline
\hline
\end{tabular}
\end{table}

\subsubsection{Counter Example using Graph Cuts}
Consider a weighted undirected graph $G = (V,E)$ comprising of a set $V$ of vertices together with a set $E = \{(u,v): u,v \in V\}$ of edges. Let $w: E \rightarrow \mathbb{R}^+$ be a positive weight function on the edges. Define a cut function $f: 2^V \rightarrow \mathbb{R}^+$ by $f(S \subseteq V) \equiv \sum\limits_{ u \in S,v \in V-S} w((u,v))$. Given a $t \in (0,1)$, the problem of actionable clustering is to partition $V$ into $k$ non-trivial clusters $\mathcal{C}^{\ast} = \{C_1^{\ast},\ldots,C_k^{\ast}\}$ where
\begin{equation*}
\mathcal{C}^{\ast}= \argmin\limits_{\mathcal{C}} \sum_{i=1}^{k} f(C_i) \text{ subject to: } \exists C_i\in \mathcal{C}~~\text{with}~~~|C_i| \ge t |V|.
\end{equation*}
It is well known that  the cut function $f$ is both symmetric and submodular \cite{Nagamochi92,Nagamochi94,quey}. Here we have only one group $g=V$. 

To construct the required counterexample we let the graph $G$ to be a \emph{disjoint union} of a complete graph $G_1 = (V_1,E_1)$ and a tree $G_2 = (V_2,E_2)$ with $|V_1| = (1-t)|V|$ and $|V_2| = t|V|$. For an arbitrary small $\epsilon > 0$ let the edges in $E_1$ and $E_2$ be weighted $1/\epsilon$ and $\epsilon$ respectively. To make the problem more concrete set $|V| = 100, t = 0.51$ and $k=10$. Further, let $a$ and $b$ be two nodes in $G_2$ of degree $1$.

By construction, the first iteration the algorithm in Table~\ref{tab:algorithm} will partition $V = (A, B)$ where $A = V_1$ and $B=V_2$. The first split cost $f(A)+f(B) = 0$ and that $A \in \mathcal{I}(V)$ and $B$ upholds the cardinality constraint as $|B| = 51 = t|V|$. Any further division of the set $B$ will result in an infeasible solution as no cluster would be large enough to be actionable. So the algorithm would repeatedly split $A$ into $k-1 = 9$ clusters, namely $\{A_1, A_2,\ldots,A_9\}$ and if we let $A_{k} = B$, the total cost would amount to $\sum_{i=}^{k} A_i$. As $G_1$ is a complete graph and the edges in $E_1$ are heavily weighted by design, even the optimal $k-1$ partition of A---not necessarily the clusters outputted by our algorithm---is expensive. However, the more efficient and perhaps the optimal solution will be to club $V_1$ with nodes $\{a,b\}$ and set $C_1 = V_1 \bigcup \{a,b\}$ making it just large enough to be actionable ($|C_1| = 51$) and then optimally partition $V_2-\{a,b\}$ into $9$ clusters $C_2,\ldots,C_{10}$. Then for any $\alpha >0$ we can choose an $\epsilon$ such that $\sum_{i=1}^{k}f(A_i) > \alpha \sum_{i=1}^{k}f(C_i)$ as the former clusters are obtained by cutting through costly edges of weight $1/\epsilon$ and the latter is a resultant of subdividing across inexpensive edges of mass $\epsilon$. Hence the proposed algorithm cannot promise any performance guarantees. The sequential splitting of $V$ cannot be greedily performed as in GSA.

\subsection{Strategy 2: Parallel Splitting}
An alternative to sequentially splitting is to simultaneously split and in a single shot obtain a $k$ partition. Algorithm 2 in \cite{heredit} offers this possibility as it outputs all minimal disjoint optimal solutions under the hereditary family constraint\footnote{The argument for why all minimal optimal solutions are disjoint is given in the cited paper.}. The obvious issue with this approach is that the union of $k-1$ of these sets may not provide a feasible $k$ partition. Moreover, we may not get $k-1$ solutions but much fewer in which case we would not have a $k$ partition. A possibility is to relax our constraint and say we want $\le k$ clusters. However, since our function $f$ is submodular the optimal solution would then be not splitting at all i.e. the original dataset $D$. We could always say that we want an $m$ partition where $m\le k$ and is as close as possible to $k$ at the expense of the definition sounding a little artificial.

Nonetheless, if we do obtain a feasible $k$ (or $m\le k$) partition by this approach it can be shown to satisfy the better than 2 approximation guarantee.
\begin{lem}
Given a feasible $k$ partition $\{C_1,...,C_k\}$ obtained using algorithm 2 in \cite{heredit} and the optimal (feasible) $k$ partition given by $\{C^{\ast}_1,...,C^{\ast}_k\}$ we have \[\sum_{i=1}^{k}f(C_i) \le 2\left(1-\frac{1}{k}\right)\sum_{i=1}^{k}f(C^{\ast}_i).\]
\end{lem}
\begin{proof}
Without loss of generality assume $C_i$ $\forall i\in\{1,...,k-1\}$ are the minimal optimal solutions and that $f(C^{\ast}_k)=\max\limits_{i\in\{1,...,k\}} f(C^{\ast}_i)$ then,
\begin{align*}
\sum_{i=1}^{k}f(C_i) &= \sum_{i=1}^{k-1}f(C_i)+ f\left(\bigcup\limits_{i=1}^{k-1}C_i\right) \\
&\le 2\sum_{i=1}^{k-1}f(C_i) \le 2\sum_{i=1}^{k-1}f(C^{\ast}_i) \le 2\left(1-\frac{1}{k}\right)\sum_{i=1}^{k}f(C^{\ast}_i).
\end{align*}
\end{proof}
The first equality is due to symmetry. The first inequality is because of submodularity. The second inequality is because the $C_i$ $\forall i\in\{1,...,k-1\}$ are optimal binary partitions. The final inequality is because $f(C^{\ast}_k)$ has the highest value amongst the $C^{\ast}_i$ and hence, $f(C^{\ast}_k)\ge \frac{1}{k}\sum_{i=1}^{k}f(C^{\ast}_i)$.

\section{Discussion}
In this paper, we saw that sequentially splitting until we obtain a k-partition by extending GSA or simultaneously obtaining minimal disjoint solutions may not lead to a $k$-partition with approximation guarantees. Although for $k=2$ we have a globally optimal solution that runs in polynomial time and if relax our requirement to obtaining $\le k$ clusters we are able to get a solution with similar approximation guarantees as in the unconstrained setting, getting an exact $k$-partition with approximation guarantees is still an open question.

Less greedy strategies than GSA that take into account $k$ and split accordingly so as to keep the option of splitting any cluster viable until we have a $k$-partition could be the answer. However, doing this in a way that each splitting step can be bounded is the main challenge. Following the approach in \cite{Chekuri11}, another possibility could be to relax the problem using Lovasz extension and obtain a bounded solution using an intelligent rounding mechanism.

\bibliographystyle{abbrv}
\bibliography{subclust} 
\end{document}